\documentclass[reqno]{amsart}
\usepackage{fullpage,amsfonts,amssymb,amsthm,mathrsfs,graphicx,color,framed}
\usepackage{cancel,bm}

\usepackage[colorlinks=true, pdfstartview=FitV, linkcolor=blue, citecolor=blue, urlcolor=blue]{hyperref}
\definecolor{shadecolor}{rgb}{0.95, 0.95, 0.86}

\renewcommand{\d}{{\mathrm d}}

\newcommand{\e}{\mathrm{e}}

\numberwithin{equation}{section}

\newtheorem{theo}{Theorem}[section]

\newtheorem{rem}[theo]{Remark}

\newtheorem{prop}[theo]{Proposition} 
\newtheorem{cor}[theo]{Corollary}

\begin{document}

\title[Short distance asymptotics]{A short note on the scaling function constant problem in the two-dimensional Ising model}

\author{Thomas Bothner}
\address{Department of Mathematics, University of Michigan, 2074 East Hall, 530 Church Street, Ann Arbor, MI 48109-1043, United States}
\email{bothner@umich.edu}



\thanks{The author is grateful to C. Tracy and A. Its for stimulating discussions about this project. This work is supported by the AMS and the Simons Foundation through a travel grant.}

\begin{abstract} We provide a simple derivation of the constant factor in the short-distance asymptotics of the tau-function associated with the $2$-point function of the two-dimensional Ising model. This factor was first computed by C. Tracy in \cite{T} via an exponential series expansion of the correlation function. Further simplifications in the analysis are due to Tracy and Widom \cite{TW} using Fredholm determinant representations of the correlation function and Wiener-Hopf approximation results for the underlying resolvent operator. Our method relies on an action integral representation of the tau-function and asymptotic results for the underlying Painlev\'e-III transcendent from \cite{MTW}. 
\end{abstract}

\date{\today}
\maketitle
\section{Introduction and statement of results}\label{sec:11}
This paper is concerned with the short distance expansion of the spin-spin correlation functions $\langle \sigma_{00}\sigma_{MN}\rangle$ for the 2D Ising model in the scaling limit analyzed by Barouch, McCoy, Tracy and Wu \cite{BMTW}.
\subsection{Definition of the model} The states of the 2D Ising model on a rectangular lattice are random spins $\sigma_{ij}=\pm 1$ at sites $(i,j)\in\mathbb{Z}^2$. As usual in Gibbs's statistical mechanics, one first analyzes the infinite system in a finite box $\Lambda=\{(i,j)\in\mathbb{Z}^2:\,0\leq i\leq M,0\leq j\leq N\}$ and then passes to the thermodynamic limit $\Lambda\uparrow\mathbb{Z}^2$. For finite $\Lambda$ the nearest-neighbor interaction energy of a configuration $\sigma=(\sigma_{ij})$ equals
\begin{equation*}
	\mathcal{E}_{\Lambda}(\sigma)=-J_1\sum_{(i,j)\in\Lambda}\sigma_{ij}\sigma_{ij+1}-J_2\sum_{(i,j)\in\Lambda}\sigma_{ij}\sigma_{i+1j},\ \ \ \ \ J_i>0,
\end{equation*}
and we assume periodic boundary conditions for simplicity, $\sigma_{i+M+1j}=\sigma_{ij},\sigma_{ij+N+1}=\sigma_{ij}$. The objects of principal physical interest are the Ising correlations, i.e. for a finite set $A\subset\Lambda$ the expectations
\begin{equation}\label{Ising}
	\left\langle\prod_{(i,j)\in A}\sigma_{ij}\right\rangle=\sum_{\sigma}\prod_{(i,j)\in A}\sigma_{ij}\mu_{\Lambda}(\sigma);\hspace{1cm}
	\mu_{\Lambda}(\sigma)=\frac{1}{Z_{\Lambda}}\e^{-\beta\mathcal{E}_{\Lambda}(\sigma)},
\end{equation}
where $\mu_{\Lambda}(\sigma)$ is the Gibbs measure of a given configuration $\sigma$ in $\Lambda$ at inverse temperature $\beta=\frac{1}{k_{\textnormal{B}}T}>0$ and $Z_{\Lambda}$ serves as normalization. It is a well known fact that the limiting behavior of the Ising correlations \eqref{Ising} is temperature dependent, for instance (cf. \cite{LM,P,MW}) there exists a critical value $T_c>0$ determined via
\begin{equation*}
	\sinh(2\beta_c J_1)\sinh(2\beta_c J_2)=1,\ \ \ \ \beta_c=\beta(T_c),
\end{equation*}
such that
\begin{equation}\label{trans}
	\lim_{\Lambda\uparrow\mathbb{Z}^2}\langle\sigma_{00}\rangle^+>0\ \ \textnormal{for}\ T<T_c\hspace{1cm}\textnormal{and}\hspace{1cm}\lim_{\Lambda\uparrow\mathbb{Z}^2}\langle\sigma_{00}\rangle^+=0\ \ \textnormal{for}\ T>T_c,
\end{equation}
where $\langle\cdot\rangle^+$ denotes the average \eqref{Ising} with all boundary spins equal to $+1$. This sharp transition between order and disorder marks the existence of a phase transition in the 2D Ising model and our interest here lies on the behavior of the two-point function $\langle\sigma_{00}\sigma_{MN}\rangle$ in the {\it massive scaling limit} $\Lambda\uparrow\mathbb{Z}^2,T\uparrow\downarrow T_c$ as studied in \cite{BMTW}.
\subsection{Scaling theory of $\langle\sigma_{00}\sigma_{MN}\rangle$} Let $z_1=\tanh(\beta J_1)\in(0,1),z_2=\tanh(\beta J_2)\in(0,1)$ and note that
\begin{equation*}
	z_1z_2+z_1+z_2=1+2\big(J_1(1-z_{1c})+J_2(1-z_{2c})\big)(\beta-\beta_c)+\mathcal{O}\left((\beta-\beta_c)^2\right),\ \ \beta\rightarrow\beta_c,
\end{equation*}
with $z_{jc}=z_j(T_c)$. With $R>0$ denoting the spatial distance
\begin{equation}\label{space}
	R=\left(\sqrt{\frac{z_1(1-z_2^2)}{z_2(1-z_1^2)}}M^2+\sqrt{\frac{z_2(1-z_1^2)}{z_1(1-z_2^2)}}N^2\right)^{\frac{1}{2}},
\end{equation}
we recall the following fundamental result.
\begin{theo}[Barouch-McCoy-Tracy-Wu \cite{BMTW}]\label{cent} Let $T\rightarrow T_c$ and $R\rightarrow\infty$ such that
\begin{equation*}
	\lim_{\substack{T\rightarrow T_c \\ R\rightarrow\infty}}\frac{\big|z_1z_2+z_1+z_2-1\big|}{\sqrt[4]{z_1z_2(1-z_1^2)(1-z_2^2)}}\,R=t
\end{equation*}
exists with $t\in(0,+\infty)$. Then
\begin{equation}\label{PIII}
	\lim_{\substack{T\downarrow\uparrow T_c\\ R\rightarrow\infty}}R^{\frac{1}{4}}\langle\sigma_{00}\sigma_{MN}\rangle=(2t)^{\frac{1}{4}}\big(\sinh(2\beta_cJ_1)+\sinh(2\beta_cJ_2)\big)^{\frac{1}{8}}\tau_{\pm}\left(t,\frac{1}{\pi}\right),
\end{equation}
with
\begin{equation*}
	\tau_{\pm}(t,\lambda)=\exp\left[\frac{1}{4}\int_t^{\infty}\left(\sinh^2\psi(s,\lambda)-\left(\frac{\d\psi}{\d s}(s,\lambda)\right)^2\right)s\,\d s\right]\begin{cases}\sinh\frac{1}{2}\psi(t,\lambda),&T\downarrow T_c\ (+)\smallskip\\ \cosh\frac{1}{2}\psi(t,\lambda),&T\uparrow T_c\ (-)\end{cases}.
\end{equation*}
The function $\psi=\psi(t,\lambda),\lambda\pi\in[0,1]$ is a distinguished solution to the radial sinh-Gordon equation
\begin{equation}\label{e:1}
	\frac{\d^2\psi}{\d t^2}+\frac{1}{t}\frac{\d\psi}{\d t}=\frac{1}{2}\sinh(2\psi),
\end{equation}
uniquely determined by the boundary condition 
\begin{equation}\label{plus}
	\psi(t,\lambda)\sim2\lambda K_0(t),\ \ \ t\rightarrow+\infty,
\end{equation} 
in terms of the modified Bessel function $K_0(z)$, cf. \cite{NIST}.
\end{theo}
The scaling limits $\tau_{\pm}(t,\lambda)$ in \eqref{PIII} are equivalently expressed in terms of a Painlev\'e-III transcendent with parameters $(\alpha,\beta,\gamma,\delta)=(0,0,1,-1)$ by recalling that $u(x,\lambda)=\e^{-\psi(t,\lambda)},t=2x$ solves
\begin{equation*}
	\frac{\d^2 u}{\d x^2}=\frac{1}{u}\left(\frac{\d u}{\d x}\right)^2-\frac{1}{x}\frac{\d u}{\d x}+u^3-\frac{1}{u}.
\end{equation*}
Famously, \eqref{PIII} (together with the later works of Ablowitz and Segur \cite{AS}) was central to all subsequent developments in modern Painlev\'e special function theory, in particular the first rigorous solution of a Painlev\'e connection problem was obtained by McCoy, Tracy and Wu in their subsequent work \cite{MTW}: For the one-parameter family of solutions $\psi(t,\lambda)$ to \eqref{e:1} with $t\in(0,+\infty)$ and $\lambda\pi\in[0,1]$ subject to the boundary condition \eqref{plus}, we have,
\begin{equation}\label{e:4}
	\psi(t,\lambda)=-\sigma \ln t-\ln B+\mathcal{O}\left(t^{2(1-\sigma)}\right),\ \ \ t\downarrow 0,
\end{equation}
uniformly for $\lambda\pi\in[0,1)$ chosen from compact subsets. The error term in \eqref{e:4} is $t$- and $\lambda$-differentiable and the coefficients $(\sigma,B)$ are the following explicit functions of the parameter $\lambda$,
\begin{equation}\label{e:5}
	\sigma=\sigma(\lambda)=\frac{2}{\pi}\arcsin(\lambda\pi)\in[0,1);\ \ \ \ \ \ \ \ \ B=B(\lambda)=2^{-3\sigma}\frac{\Gamma(\frac{1}{2}(1-\sigma))}{\Gamma(\frac{1}{2}(1+\sigma))},
\end{equation}
in terms of Euler's gamma function $\Gamma(z)$. In fact, $\psi(\cdot,\lambda)$ is smooth on the positive real axis for all $\lambda\pi\in[0,1]$ and we have in addition to \eqref{e:4},
\begin{equation}\label{add}
	\psi\left(t,\frac{1}{\pi}\right)=-\ln t-\ln\Big(-\frac{1}{2}\left\{\ln\left(\frac{t}{8}\right)+\gamma_E\right\}\Big)+\mathcal{O}\left(t^4\ln ^2t\right),\ \ \ \ \ t\downarrow 0,
\end{equation}
in terms of Euler's constant $\gamma_E$. The boundary behavior of $\psi(t,\lambda)$ allows us to compute the long- and short-distance expansions of $\tau_{\pm}(t,\lambda)$ and we can then compare these results to the scaling hypothesis of the two-point function, cf \cite{BMTW}.
\subsection{Tau-function connection problem} The radial sinh-Gordon equation \eqref{e:1} admits the Hamiltonian formulation
\begin{equation}\label{e:8}
	\frac{\d q}{\d t}=\frac{\partial H}{\partial p},\ \ \ \ \ \frac{\d p}{\d t}=-\frac{\partial H}{\partial q},\ \ \ \ \ \ \ H=H(q,p,t)=\frac{t}{2}\sinh^2 q-\frac{p^2}{2t},
\end{equation}
with the identification $q=q(t,\lambda)\equiv\psi(t,\lambda)$. This allows us to rewrite the above formula for $\tau_{\pm}(t,\lambda)$ as
\begin{equation}\label{e:9}
	\tau_{\pm}(t,\lambda)=\exp\left[\frac{1}{2}\int_t^{\infty}H(q,p,s)\,\d s\right]\begin{cases}\sinh\frac{1}{2}q,&(+)\smallskip\\ \cosh\frac{1}{2}q,&(-)\end{cases}.
\end{equation}
In short, $\tau_{\pm}(t,\lambda)$ (up to the $\sinh$ and $\cosh$ factors) is a tau-function, cf. \cite{JMU}, for \eqref{e:1} and the underlying Barouch-McCoy-Tracy-Wu family of solutions $\psi=\psi(t,\lambda)$. The problem of determining the complete asymptotic description of $\tau_{\pm}(t,\lambda)$ as $t\downarrow 0$ provided the same description is given as $t\rightarrow+\infty$ (or vice versa) is known as tau-function connection problem. Parts of this problem are easy, indeed, using \eqref{plus}, \eqref{e:4}, \eqref{e:5} and \eqref{add} in \eqref{e:9}, we obtain at once
\begin{equation*}
	\tau_{\pm}(t,\lambda)\sim \begin{cases}\lambda\sqrt{\frac{\pi}{2t}}\,\e^{-t},&(+)\bigskip\\ 1+\frac{\pi\lambda^2}{8t^2}\,\e^{-2t},&(-)\end{cases},\ \ \ \ t\rightarrow+\infty,\ \ \ \ \lambda\pi\in\left[0,1\right],
\end{equation*}
and, 
\begin{equation}\label{e:6}
	\tau_{\pm}(t,\lambda)\sim A(\lambda)t^{\frac{\sigma}{4}(\sigma-2)},\ \ \ \ t\downarrow 0,\ \ \ \ \lambda\pi\in\left(0,1\right],\ \ \ \ \ \sigma=\sigma(\lambda)=\frac{2}{\pi}\arcsin(\lambda\pi)\in(0,1],
\end{equation}
where $A(\lambda)$ is $t$-independent. However, obtaining a simple, closed form expression for $A(\lambda)$ is challenging when working with \eqref{e:9} only. We quote from \cite{MW}, page $420$:\bigskip
\begin{quote}
``The solutions \eqref{e:4} and \eqref{add} of the Painlev\'e connection problem are not sufficient to compute the transcendental constant $A(\lambda)$ in \eqref{e:6}. In principle this constant is obtained by the integral of the Painlev\'e function in \eqref{e:1} and should follow from the defining differential equation and boundary condition, but in practice such a derivation has never been found.''\bigskip
\end{quote}
It is the purpose of this note to provide such a derivation. Before presenting our method, we emphasize that Tracy \cite{T} computed $A(\lambda)$ in $1991$  (thus solving the tau-function connection problem) through an infinite series representation of $\tau_{\pm}(t,\lambda)$. His result is as follows.
\begin{theo}[Tracy \cite{T}] Let $s=\frac{1}{2}(1-\sigma)\in[0,\frac{1}{2})$ with $\sigma=\sigma(\lambda)$ as in \eqref{e:6}, then
\begin{equation}\label{e:7}
	A(\lambda)=\e^{3\zeta'(-1)-(3s^2+\frac{1}{6})\ln 2}\big(G(1+s)G(1-s)\big)^{-1},
\end{equation}
in terms of the Riemann zeta function $\zeta(z)$ and Barnes-G function $G(z)$, cf. \cite{NIST}.
\end{theo}
We will follow here a different route that relies on a novel action integral formula for $\tau_{\pm}(t,\lambda)$. This formula allows us to compute $A(\lambda)$ directly from the boundary behavior \eqref{e:4}, see Section \ref{sec:12} below, and our method does not require any additional integrable structures (such as Fredholm determinant formul\ae\, which were used in \cite{TW}, compare our discussion in Section \ref{sec:13}).\bigskip

Formula \eqref{e:7} is useful for the scaling hypothesis of the 2D Ising model: A special case of \eqref{e:7} appeared first in \cite{W},(4.14),(5.7),(7.2) for $\lambda\pi=1$. In detail, Wu showed that at $T=T_c$,
\begin{equation}\label{tw}
	\langle\sigma_{00}\sigma_{0N}\rangle\Big|_{T=T_c}=\e^{3\zeta'(-1)+\frac{1}{12}\ln 2}N^{-\frac{1}{4}}\sqrt[4]{\frac{1+\tanh^2(\beta_c J_1)}{1-\tanh^2(\beta_c J_1)}}\left(1+\mathcal{O}\left(N^{-2}\right)\right),\ \ \ N\rightarrow\infty,
\end{equation}
and it was conjectured that the numerical constant in the leading order of \eqref{tw} equals
\begin{equation*}
	\lim_{\substack{T\downarrow\uparrow T_c\\ R\rightarrow\infty}}R^{\frac{1}{4}}\langle\sigma_{00}\sigma_{0N}\rangle\bigg|_{t=0}=\lim_{t\downarrow 0}\left\{(2t)^{\frac{1}{4}}\big(\sinh(2\beta_c J_1)+\sinh(2\beta_c J_2)\big)^{\frac{1}{8}}\tau_{\pm}\left(t,\frac{1}{\pi}\right)\right\}.
\end{equation*}
In order to see this we first manipulate \eqref{tw} with the help of \eqref{space} (here $M=0$) and the definition of $T_c$,
\begin{equation*}
	R^{\frac{1}{4}}\langle\sigma_{00}\sigma_{0N}\rangle\Big|_{T=T_c}=\e^{3\zeta'(-1)+\frac{1}{12}\ln 2}\big(\sinh(2\beta_c J_1)+\sinh(2\beta_cJ_2)\big)^{\frac{1}{8}}\left(1+\mathcal{O}\left(N^{-2}\right)\right),\ \ N\rightarrow\infty.
\end{equation*}
But from \eqref{e:6},
\begin{equation*}
	\lim_{t\downarrow 0}\left\{(2t)^{\frac{1}{4}}\big(\sinh(2\beta_c J_1)+\sinh(2\beta_c J_2)\big)^{\frac{1}{8}}\tau_{\pm}\left(t,\frac{1}{\pi}\right)\right\}=2^{\frac{1}{4}}\big(\sinh(2\beta_cJ_1)+\sinh(2\beta_cJ_2)\big)^{\frac{1}{8}}A\left(\frac{1}{\pi}\right),
\end{equation*}
and with \eqref{e:7},
\begin{equation*}
	2^{\frac{1}{4}}A\left(\frac{1}{\pi}\right)=\e^{3\zeta'(-1)+\frac{1}{12}\ln 2},
\end{equation*}
so equality of the constants indeed follows. As emphasized in \cite{T} for the symmetrical lattice (i.e. $J_1=J_2$) the above observation closes a small gap in the proof of the scaling hypothesis of the $2$-point function in the works of Barouch, McCoy, Tracy and Wu.
\subsection{Outline of paper} In Section \ref{sec:12} we compute \eqref{e:7} by rewriting the Hamiltonian integral \eqref{e:9} as action integral plus explicit terms, i.e. terms without any integrals. Our formula is particularly useful for asymptotic analysis since we can shift $t$-integration in the action integral to $\lambda$-integration, compare Corollary \ref{nice} below. After that we simply substitute \eqref{e:4} into \eqref{e:11} and express the remaining $\lambda$-integrals as Barnes-G functions. This completes our proof of \eqref{e:7}. After that, in Section \ref{gen}, we briefly discuss an extension of our method to a $\nu$-generalization of $\tau_{\pm}(t,\lambda)$ that appeared in \cite{MTW}. In Section \ref{sec:13} we conclude with a brief discussion of other recent occurrences of action integral formul\ae\,in exactly solvable models.
\section{Proof of \eqref{e:7} via action integral formula}\label{sec:12}
Our proof hinges crucially upon the following identity.
\begin{prop} Let $q=q(t,\lambda)$ and $p=p(t,\lambda)$ solve \eqref{e:8} subject to \eqref{plus} for $t\in(0,\infty)$ and $\lambda\pi\in[0,1]$. Then
\begin{equation}\label{e:10}
	\int_t^{\infty}H(q,p,s)\,\d s=-tH(q,p,t)+S(t,\lambda),
\end{equation}
where $S$ is the classical action 
\begin{equation*}
	S(t,\lambda)=\int_t^{\infty}\left(p\frac{\d q}{\d s}-H(q,p,s)\right)\,\d s.
\end{equation*}
\end{prop}
\begin{proof} We $t$-differentiate the right hand side in \eqref{e:10} and use \eqref{e:8},
\begin{equation*}
	\frac{\d}{\d t}\Big[-tH(q,p,t)+S(t,\lambda)\Big]=-H-t\frac{\partial H}{\partial t}-p\frac{\d q}{\d t}+H=-\frac{t}{2}\sinh^2q+\frac{p^2}{2t}=-H.
\end{equation*}
Thus both sides in \eqref{e:10} can only differ by a $t$-independent additive term, but since
\begin{equation*}
	q(t,\lambda)\sim2\lambda K_0(t)\sim2\lambda\sqrt{\frac{\pi}{2t}}\,\e^{-t},\ \ \ \ t\rightarrow+\infty,
\end{equation*}
both sides in \eqref{e:10} decay exponentially fast at $t=+\infty$, i.e. \eqref{e:10} follows.
\end{proof}
\begin{rem} In \cite{T0}, section $6.1$, Tracy observed that the Hamiltonian integral in \eqref{e:9} ``looks almost like an action''. Formula \eqref{e:10} gives us closure on this matter: the antiderivative of the Hamiltonian in \eqref{e:9} is an action integral modulo explicit terms, i.e. terms without integrals.
\end{rem}
In order to appreciate the right hand side in \eqref{e:10}, we note that
\begin{cor}\label{nice} For any $t\in(0,\infty)$ and $\lambda\pi\in[0,1]$,
\begin{equation*}
	S(t,\lambda)=-\int_0^{\lambda}p\frac{\partial q}{\partial\lambda'}\,\d\lambda'.
\end{equation*}
\end{cor}
\begin{proof} We $\lambda$-differentiate $S(t,\lambda)$, use \eqref{e:8} and integrate by parts,
\begin{equation*}
	\frac{\partial S}{\partial\lambda}=\int_t^{\infty}\left(\frac{\partial p}{\partial\lambda}\frac{\d q}{\d s}+p\frac{\partial}{\partial\lambda}\frac{\d q}{\d s}\right)\,\d s=p\frac{\partial q}{\partial\lambda}\bigg|_{s=t}^{\infty}=-p\frac{\partial q}{\partial\lambda}.
\end{equation*}
But $q(t,0)\equiv 0$, which completes the proof after $\lambda$-integration.
\end{proof}
Summarizing, we have derived the identity
\begin{equation}\label{e:11}
	\tau_{\pm}(t,\lambda)=\exp\left[-\frac{t}{2}H(q,p,t)-\frac{1}{2}\int_0^{\lambda}p\frac{\partial q}{\partial\lambda'}\,\d\lambda'\right]\begin{cases}\sinh\frac{1}{2}q,&(+)\smallskip\\ \cosh\frac{1}{2}q,&(-)\end{cases},\ \ \ t\in(0,\infty),\ \ \lambda\pi\in\left[0,1\right].
\end{equation}
Since expansion \eqref{e:6} holds true for all $\lambda\pi\in(0,1]$ it is sufficient to derive \eqref{e:7} for $\lambda\pi\in(0,1)$ and we set out to achieve this by means of \eqref{e:11} and \eqref{e:4}. First, 
\begin{equation}\label{e:12}
	-\frac{t}{2}H(q,p,t)=\frac{\sigma^2}{4}+\mathcal{O}\left(t^{2(1-\sigma)}\right),\ \ \ t\downarrow 0,
\end{equation}
which follows from straightforward $t$-differentiation of \eqref{e:4} and substitution into \eqref{e:8}. Second,
\begin{equation}\label{e:13}
	-\frac{1}{2}\int_0^{\lambda}p\frac{\partial q}{\partial\lambda'}\,\d\lambda'=\frac{\sigma^2}{4}\ln t+\underbrace{\frac{1}{2}\int_0^{\lambda}\sigma(\lambda')\frac{\partial}{\partial\lambda'}\ln B(\lambda')\,\d\lambda'}_{=L(\lambda)}+\mathcal{O}\left(t^{2(1-\sigma)}\ln t\right),
\end{equation}
again from $\lambda$-differentiation of \eqref{e:4}. In order to express the remaining integral $L(\lambda)$ in \eqref{e:13} as Barnes-G function we rely on, cf. \cite{NIST},
\begin{equation}\label{barnes}
	\int_0^z\ln\Gamma(1+x)\,\d x=\frac{z}{2}\ln(2\pi)-\frac{z}{2}(z+1)+z\ln\Gamma(1+z)-\ln G(1+z),\ \ \ z\in\mathbb{C}:\ \Re z>-1,
\end{equation}
the functional equations $\Gamma(1+z)=z\Gamma(z), G(1+z)=\Gamma(z)G(z)$ and the special values
\begin{equation}\label{spec}
	\Gamma\left(\frac{1}{2}\right)=\sqrt{\pi},\ \ \ \ \ \ \ 2\ln G\left(\frac{1}{2}\right)=3\zeta'(-1)-\frac{1}{2}\ln\pi+\frac{1}{12}\ln 2.
\end{equation}
Indeed,
\begin{equation}\label{e:14}
	L(\lambda)=-\frac{3\sigma^2}{4}\ln 2-\frac{1}{2}\int_0^{\lambda}\sigma(\lambda')\frac{\partial}{\partial\lambda'}\ln s(\lambda')\,\d\lambda'+\frac{1}{2}\int_0^{\lambda}\sigma(\lambda')\frac{\partial}{\partial\lambda'}\ln\frac{\Gamma(1+s(\lambda'))}{\Gamma(1-s(\lambda'))}\,\d\lambda'
\end{equation}
where $s=s(\lambda)=\frac{1}{2}(1-\sigma(\lambda))$. The first integral in \eqref{e:14} is elementary,
\begin{equation}\label{e:15}
	-\frac{1}{2}\int_0^{\lambda}\sigma(\lambda')\frac{\partial}{\partial\lambda'}\ln s(\lambda')\,\d\lambda'=-\frac{1}{2}\ln s-\frac{\sigma}{2}-\frac{1}{2}\ln 2,
\end{equation}
and for the second we use \eqref{barnes}, the functional equations and \eqref{spec},
\begin{align}\label{e:16}
	\frac{1}{2}\int_0^{\lambda}\sigma(\lambda')\frac{\partial}{\partial\lambda'}\ln\frac{\Gamma(1+s(\lambda'))}{\Gamma(1-s(\lambda'))}\,\d\lambda'=&\,\frac{1}{2}\ln s+\frac{1}{2}\ln\frac{\Gamma(\frac{1}{2}(1-\sigma))}{\Gamma(\frac{1}{2}(1+\sigma))}-s^2-\ln\big(G(1+s)G(1-s)\big)\nonumber\\
	&+\frac{1}{4}+\frac{7}{12}\ln 2+3\zeta'(-1).
\end{align}
Now combine \eqref{e:12}, \eqref{e:13}, \eqref{e:14}, \eqref{e:15} and \eqref{e:16},
\begin{align*}
	-\frac{t}{2}H(q,p,t)-\frac{1}{2}\int_0^{\lambda}p\frac{\partial q}{\partial\lambda'}\,\d\lambda'=&\,\frac{\sigma^2}{4}\ln t-\frac{3\sigma^2}{4}\ln 2+\frac{1}{12}\ln 2+\frac{1}{2}\ln\frac{\Gamma(\frac{1}{2}(1-\sigma))}{\Gamma(\frac{1}{2}(1+\sigma))}\\
	&-\ln\big(G(1+s)G(1-s)\big)+3\zeta'(-1)+\mathcal{O}\left(t^{2(1-\sigma)}\ln t\right),\ \ \ t\downarrow 0.
\end{align*}
Since also
\begin{equation*}
	\begin{cases}\sinh\frac{1}{2}q,&(+)\smallskip\\ \cosh\frac{1}{2}q,&(-)\end{cases}=\frac{1}{2}t^{-\frac{\sigma}{2}}B^{-\frac{1}{2}}\left(1+\mathcal{O}\left(\max\{t^{2(1-\sigma)},t^{\sigma}\}\right)\right),\ \ t\downarrow 0,
\end{equation*}
we obtain all together, as $t\downarrow 0$ and $\lambda\pi\in(0,1)$ is fixed,
\begin{equation*}
	\tau_{\pm}(t,\lambda)=\e^{3\zeta'(-1)+(\frac{3\sigma}{2}-\frac{3\sigma^2}{4}-\frac{11}{12})\ln 2}\left(G(1+s)G(1-s)\right)^{-1}t^{\frac{\sigma}{4}(\sigma-2)}\left(1+\mathcal{O}\left(\max\{t^{2(1-\sigma)}\ln t,t^{\sigma}\}\right)\right).
\end{equation*}
This expansion matches exactly \eqref{e:6}, \eqref{e:7} with $s=\frac{1}{2}(1-\sigma)$ and thus completes our derivation of $A(\lambda)$.
\section{$\nu$-generalization of \eqref{e:9}}\label{gen}
In \cite{MTW},(1.14b) the following $\nu$-generalization of $\tau_{\pm}(t,\lambda)$ appears,
\begin{align}
	\tau_{\pm}(t,\lambda,\nu)=&\,\exp\left[\frac{1}{4}\int_t^{\infty}\left\{\sinh^2\psi(s,\lambda,\nu)-\left(\frac{\d\psi}{\d s}(s,\lambda,\nu)\right)^2+\frac{4\nu}{s}\sinh^2\left(\frac{1}{2}\psi(s,\lambda,\nu)\right)\right\}\,s\,\d s\right]\label{g:1}\\
	&\,\times\begin{cases}\sinh\frac{1}{2}\psi(t,\lambda,\nu),&(+)\\ \cosh\frac{1}{2}\psi(t,\lambda,\nu),&(-)\end{cases},\ \ \ \ \ \ t>0,\nonumber
\end{align}
where $\psi=\psi(t,\lambda,\nu),t\in(0,+\infty),\lambda\pi\in[0,1],\Re\nu>-\frac{1}{2}$ solves the $\nu$-modified radial sinh-Gordon equation
\begin{equation}\label{nusinh}
	\frac{\d^2\psi}{\d t^2}+\frac{1}{t}\frac{\d\psi}{\d t}=\frac{1}{2}\sinh(2\psi)+\frac{2\nu}{t}\sinh\psi,
\end{equation}
subject to the boundary condition
\begin{equation}\label{psiinf}
	\psi(t,\lambda,\nu)\sim2\lambda\int_1^{\infty}\frac{\e^{-ty}}{\sqrt{y^2-1}}\left(\frac{y-1}{y+1}\right)^{\nu}\,\d y,\ \ \ \ t\rightarrow+\infty.
\end{equation}
For the 2D-Ising model only $\nu=0$ is of importance, hence the underlying tau-function connection problem for \eqref{g:1} has received no attention in the literature, to the best of our knowledge. Our goal is to sketch an adaptation of \eqref{e:10} to \eqref{g:1} which can be used in the solution of this problem. First we record the results of \cite{MTW} on the asymptotic analysis of the one-parameter family of solutions $\psi(t,\lambda,\nu)$ to \eqref{nusinh}. As $t\downarrow 0$,
\begin{equation}\label{g:2}
	\psi(t,\lambda,\nu)=-\sigma\ln t-\ln B-\ln\left[1-\frac{\nu}{B}(1-\sigma)^{-2}\,t^{1-\sigma}+B\nu(1+\sigma)^{-2}\,t^{1+\sigma}+\mathcal{O}\left(t^{2(1-\sigma)}\right)\right],
\end{equation}
which holds uniformly for $\lambda\pi\in[0,1)$ and $\Re\nu>-\frac{1}{2}$ chosen from compact subsets. The error term in \eqref{g:2} is again differentiable with respect to $t,\lambda,\nu$ and the coefficients $(\sigma,B)$ are the following explicit functions of the parameters $\lambda$ and $\nu$,
\begin{equation*}
	\sigma=\sigma(\lambda)=\frac{2}{\pi}\arcsin(\pi\lambda)\in[0,1);\hspace{1cm}B=B(\sigma,\nu)=2^{-3\sigma}\frac{\Gamma^2(\frac{1}{2}(1-\sigma))}{\Gamma^2(\frac{1}{2}(1+\sigma))}\frac{\Gamma(\nu+\frac{1}{2}(1+\sigma))}{\Gamma(\nu+\frac{1}{2}(1-\sigma))}.
\end{equation*}
Moreover, $\psi(\cdot,\lambda,\nu)$ is smooth on the positive real axis for all $\lambda\pi\in[0,1],\Re\nu>-\frac{1}{2}$ and we have in addition to \eqref{g:2},
\begin{equation}\label{g:3}
	\psi\left(t,\frac{1}{\pi},\nu\right)=-\ln\left[\frac{t}{2}\left\{\nu\ln^2t-C(\nu)\ln t+\frac{1}{4\nu}\left(C^2(\nu)-1\right)\right\}\right]+o(1),\ \ \ \ t\downarrow 0
\end{equation}
with $C(\nu)=1+2\nu\big(3\ln 2-2\gamma_E-\psi(1+\nu)\big)$ in terms of the digamma function $\psi(z)$, cf. \cite{NIST}. With \eqref{psiinf}, \eqref{g:2} and \eqref{g:3} back in \eqref{g:1} we find in turn,
\begin{equation*}
	\tau_{\pm}(t,\lambda,\nu)\sim\begin{cases}\lambda\frac{\Gamma(\nu+\frac{1}{2})}{(2t)^{\nu+\frac{1}{2}}}\,\e^{-t},&(+)\bigskip\\ 1-\frac{\lambda^2}{2}\frac{\nu\Gamma^2(\nu+\frac{1}{2})}{(2t)^{2\nu+1}}\,\e^{-2t},&(-)\end{cases},\ \ \ \ \ t\rightarrow+\infty,\ \ \ \ \lambda\pi\in[0,1],\ \ \Re\nu>-\frac{1}{2}
\end{equation*}
as well as,
\begin{equation*}
	\tau_{\pm}(t,\lambda,\nu)\sim A(\lambda,\nu)t^{\frac{\sigma}{4}(\sigma-2)},\ \ \ \ t\downarrow 0,\ \ \ \ \lambda\pi\in(0,1],\ \ \ \Re\nu>-\frac{1}{2}.
\end{equation*}
Here, $A(\lambda,\nu)$ is $t$-independent and obviously generalizes \eqref{e:7}. In accordance with \eqref{e:9} we now connect \eqref{g:1} to the tau-function theory of \cite{JMU}. First off the integrand in \eqref{g:1} does not appear to be a Hamiltonian for \eqref{nusinh} (at least we were not able to prove it), still it can be written as linear combination of two Painlev\'e-III Hamiltonians\footnote{Without using the Hamiltonian system \eqref{e:8} for the radial sinh-Gordon equation, but instead polynomial Hamiltonians for Painlev\'e-III, Okamoto \cite{O},(7),(8) has used a similar approach in the identification of $\tau_{\pm}(t,\lambda)$ as a tau-function product.}. Concretely, let $\tau(t,\lambda,\nu)$ abbreviate the exponential in \eqref{g:1} without the additional hyperbolic multipliers, then
\begin{equation}\label{g:4}
	\tau(t,\lambda,\nu)=\exp\left[\int_{\frac{t}{2}}^{\infty}\left\{\frac{x}{4u^2(x)}\left[\big(1-u^2(x)\big)^2-\left(\frac{\d u}{\d x}(x)\right)^2\right]+\frac{\nu}{2u(x)}\big(1-u(x)\big)^2\right\}\,\d x\right],
\end{equation}
where $u(x)\equiv u(x,\lambda,\nu)=\e^{-\psi(t,\lambda,\nu)},t=2x$ solves Painlev\'e-III with parameters $(\alpha,\beta,\gamma,\delta)=(2\nu,-2\nu,1,-1)$,
\begin{equation}\label{p:iii}
	\frac{\d^2 u}{\d x^2}=\frac{1}{u}\left(\frac{\d u}{\d x}\right)^2-\frac{1}{x}\frac{\d u}{\d x}+\frac{1}{x}(2\nu u^2-2\nu)+u^3-\frac{1}{u}.
\end{equation}
Now put
\begin{equation}\label{H:1}
	H_1(u,v_1,x,\nu)=\frac{1}{x}\Big(u^2v_1^2-\big(xu^2-(2\nu+1)u-x\big)v_1-(2\nu+1)xu\Big)+\frac{(2\nu+1)^2}{4x}+2\nu+1,
\end{equation}
and
\begin{equation}\label{H:2}
	H_2(u,v_2,x,\nu)=\frac{1}{x}\Big(u^2v_2^2+\big(xu^2-(2\nu-1)u-x\big)v_2-(2\nu-1)xu\Big)+\frac{(2\nu-1)^2}{4x}+2\nu-1,
\end{equation}
which are both Hamiltonians for \eqref{p:iii}, cf. \cite{NIST,O}. Moreover, $\frac{\d u}{\d x}=\frac{\partial H_j}{\partial v_j},\frac{\d v_j}{\d x}=-\frac{\partial H}{\partial u},j=1,2$ yield the equations
\begin{equation}\label{vj}
	v_1=\frac{1}{2u^2}\left[x\frac{\d u}{\d x}+xu^2-(2\nu+1)u-x\right],\ \ \ \ \ v_2=\frac{1}{2u^2}\left[x\frac{\d u}{\d x}-xu^2+(2\nu-1)u+x\right],
\end{equation}
and thus, upon elimination of $v_j$ in \eqref{H:1} and \eqref{H:2}, 
\begin{equation*}
	H_1=\frac{x}{4u^2}\left[\left(\frac{\d u}{\d x}\right)^2-(1-u^2)^2\right]-\frac{2\nu+1}{2u}(1-u)^2,\ \ H_2=\frac{x}{4u^2}\left[\left(\frac{\d u}{\d x}\right)^2-(1-u^2)^2\right]-\frac{2\nu-1}{2u}(1-u)^2.
\end{equation*}
Taking a linear combination of $H_1$ and $H_2$ (compare the integrand in \eqref{g:4}) we arrive at the following result.
\begin{prop}\label{product} Assume $(u,v_1,v_2)$ solve the Hamiltonian systems
\begin{equation*}
	\frac{\d u}{\d x}=\frac{\partial H_1}{\partial v_1},\ \ \ \frac{\d v_1}{\d x}=-\frac{\partial H_1}{\partial u};\hspace{2cm}\frac{\d u}{\d x}=\frac{\partial H_2}{\partial v_2},\ \ \ \frac{\d v_2}{\d x}=-\frac{\partial H_2}{\partial u};
\end{equation*}
defined in \eqref{H:1}, \eqref{H:2} for $x\in(0,+\infty),\lambda\pi\in[0,1],\Re\nu>-\frac{1}{2}$ subject to the $u$-boundary condition
\begin{equation}\label{ubound}
	u(x,\lambda,\nu)\sim\exp\left[-2\lambda\int_1^{\infty}\frac{\e^{-2xy}}{\sqrt{y^2-1}}\left(\frac{y-1}{y+1}\right)^{\nu}\,\d y\right],\ \ x\rightarrow+\infty,
\end{equation}
and the corresponding $v_j$-boundary behavior computed from \eqref{vj} and \eqref{ubound}. Then
\begin{equation}\label{g:10}
	\tau(t,\lambda,\nu)=\exp\left[-\int_{\frac{t}{2}}^{\infty}\left\{\frac{1}{2}(1-\nu)H_1(u,v_1,x,\nu)+\frac{1}{2}(1+\nu)H_2(u,v_2,x,\nu)\right\}\,\d x\right].
\end{equation}
\end{prop}
Formula \eqref{g:10} shows that $\tau(t,\lambda,\nu)$ is a product of two tau-functions, one originating from the Hamiltonian \eqref{H:1} and another one coming from \eqref{H:2}. We now state the analogue of \eqref{e:10} for the underlying Hamiltonian integrals.
\begin{prop} Under the assumptions of Proposition \ref{product},
\begin{equation}\label{g:11}
	\int_t^{\infty}H_j(u,v_j,x,\nu)\,\d x=-tH_j(u,v_j,t,\nu)-L_j(u,t,\nu)+S_j(t,\lambda,\nu),\ \ \ \ j=1,2,
\end{equation}
where $S_j$ are the classical actions
\begin{equation*}
	S_j(t,\lambda,\nu)=\int_t^{\infty}\left(v_j\frac{\d u}{\d x}-H_j(u,v_j,x,\nu)\right)\,\d x,
\end{equation*}
and we have
\begin{equation*}
	L_1(u,t,\nu)=\frac{1}{2}(2\nu+1)\ln u+\frac{1}{2}(2\nu+1)\int_t^{\infty}\frac{1}{u}(1-u)^2\,\d x,
\end{equation*}
as well as
\begin{equation*}
	L_2(u,t,\nu)=-\frac{1}{2}(2\nu-1)\ln u+\frac{1}{2}(2\nu-1)\int_t^{\infty}\frac{1}{u}(1-u)^2\,\d x.
\end{equation*}
\end{prop}
\begin{proof} We verify \eqref{g:11} by $t$-differentiation of left and right hand side using the Hamiltonian systems \eqref{H:1} and \eqref{H:2}. This straightforward computation shows that both sides in \eqref{g:11} can only differ by a $t$-independent additive term, however since from \eqref{psiinf},
\begin{equation*}
	u(x,\lambda,\nu)\sim 1-2\lambda\frac{\Gamma(\nu+\frac{1}{2})}{(4x)^{\nu+\frac{1}{2}}}\e^{-2x},\ \ \ \ \ \ \frac{\d u}{\d x}(x,\lambda,\nu)\sim 4\lambda\frac{\Gamma(\nu+\frac{1}{2})}{(4x)^{\nu+\frac{1}{2}}}\e^{-2x},\ \ \ x\rightarrow+\infty,
\end{equation*}
and from \eqref{vj},
\begin{equation*}
	v_1(x,\lambda,\nu)\sim-\frac{1}{2}(2\nu+1),\ \ \ \ \ \ v_2(x,\lambda,\nu)\sim\frac{1}{2}(2\nu-1),\ \ \ x\rightarrow+\infty,
\end{equation*}
both sides in \eqref{g:11} vanish exponentially fast at $t=+\infty$, thus the given identities follow.
\end{proof}
Observe that for any $t\in(0,\infty),\lambda\pi\in[0,1]$ and $\Re\nu>-\frac{1}{2}$, compare Corollary \ref{nice} above,
\begin{equation*}
	S_j(t,\lambda,\nu)=-\int_0^{\lambda}v_j\frac{\partial u}{\partial\lambda'}\,\d\lambda',
\end{equation*}
however the integrand in \eqref{g:10} also leads to the combination
\begin{equation*}
	\frac{1}{2}(1-\nu)L_1(u,t,\nu)+\frac{1}{2}(1+\nu)L_2(u,t,\nu)=\frac{1}{2}(1-2\nu^2)\ln u+\frac{\nu}{2}\int_t^{\infty}\frac{1}{u}(1-u)^2\,\d x,
\end{equation*}
which is not fully explicit due to the remaining $t$-integral (and this one drops out precisely for $\nu=0$). This situation is completely analogous to the author's recent work \cite{BIP}, see Theorem $1.5$, equation $(1.28)$ where a similar integral term survived. It was shown that this term (opposed to the Hamiltonian integral) admits a straightforward Riemann-Hilbert representation, see Appendix $A$, $(A.24), (A.25)$ in \cite{BIP}. This means we can compute the remaining integral explicitly from an asymptotic/nonlinear steepest descent resolution of the Riemann-Hilbert problem for Painlev\'e-III (see Section $6.1$ in \cite{BIP}, where this is worked out for a Painlev\'e-III transcendent occurring in random matrix theory). For the sake of brevity we choose not to carry out the relevant analysis here, the solution of the tau-function connection problem based on \eqref{g:10} and \eqref{g:11} will be part of a separate future publication.

\section{Closing remarks}\label{sec:13}
Previously \cite{T,TW} $A(\lambda)$ was derived from Fredholm determinant representations of the spin-spin correlation function, for instance for $T<T_c$ and $N\geq 0$ one has the form factor expansion
\begin{equation}\label{Fred:1}
	\langle\sigma_{00}\sigma_{MN}\rangle=\mathcal{M}_-^2\det\big(1+g_{MN}\upharpoonright_{L^2((-\pi,\pi), \d\theta)}\big),
\end{equation}
with the spontaneous magnetization
\begin{equation*}
	\mathcal{M}_-=\big(1-k^2\big)^{\frac{1}{8}},\ \ \ \ k=\frac{1}{\sinh(2\beta J_1)\sinh(2\beta J_2)},
\end{equation*}
and the operator kernel $g_{MN}(\theta_1,\theta_2)$ given in \cite{P}, page $142$. Carrying out the massive scaling limit in \eqref{Fred:1} one finds in turn (\cite{P}, Theorem $3.6.1$)
\begin{equation}\label{Fred:2}
	\tau_-\left(t,\frac{1}{\pi}\right)=\det\big(1-K_t^2\upharpoonright_{L^2((0,\infty),\frac{\d\lambda}{2\pi\lambda})}\big),
\end{equation}
where $K_t$ is the operator on $L^2((0,\infty),\frac{\d\lambda}{2\pi\lambda})$ with kernel
\begin{equation*}
	K_t(\lambda,\mu)=\e^{-\frac{t}{4}(\lambda+\lambda^{-1})}\frac{\lambda-\mu}{\lambda+\mu}\e^{-\frac{t}{4}(\mu+\mu^{-1})}.
\end{equation*}
The additional (operator theoretical) integrable structure \eqref{Fred:2} is extremely useful in the solution of the tau-function connection problem and a common theme in nonlinear mathematical physics as several distribution, gap or correlation functions in random matrix theory, the theory of lattice models and field theories have Fredholm (or Hankel, or Toeplitz) determinant formul\ae. In particular, nearly every isomonodromic tau-function (such as our $\tau_{\pm}(t,\lambda)$) that appears in the aforementioned fields has been analyzed asymptotically based on their operator theoretical representations. Still, these tau-functions are non-generic examples in the theory of Jimbo-Miwa-Ueno, cf. \cite{JMU}. For instance, equation \eqref{e:1} possesses a two-parameter family of solutions characterized by
\begin{equation*}
	\psi(t)=\alpha\ln t+\beta+\mathcal{O}\left(t^{2(1-|\Re\alpha|)}\right),\ \ \ t\downarrow 0,
\end{equation*}
where $\alpha\in\mathbb{C}:|\Re\alpha|<1$ and $\beta\in\mathbb{C}$ are the parameters (the Cauchy data) of the solution $\psi(t)\equiv\psi(t,\alpha,\beta)$, cf. \cite{FIKN}, Chapter $15$ (with $u(x)=2\psi(t)$ where $t=2x$ in this reference). For generic choices of $(\alpha,\beta)$ the large $t$-behavior of $\psi(t,\alpha,\beta)$ is oscillatory and thus very different from \eqref{plus}. Moreover there are no known Fredholm (or Hankel, or Toeplitz) formul\ae\,for the corresponding tau-function,
\begin{equation*}
	\frac{\d }{\d t}\ln\tau(t,\alpha,\beta)=-\frac{1}{2}H(q,p,t),\ \ \ \ q=q(t,\alpha,\beta)\equiv\psi(t,\alpha,\beta),
\end{equation*}
with Hamiltonian as in \eqref{e:8}. One way to solve the tau-function connection problem in this generic setup was recently discovered by Gamayun, Iorgov, Lisovyy, Tykhyy, Its and Prokhorov \cite{GIL,ILT,IP}. The authors propose an extension of the classical Jimbo-Miwa-Ueno $1$-form $\omega\propto H(q,p,t)\d t$ to a differential form which is closed not only on the space of isomonodromic times (in our case $t$) but on the full space of extended monodromy data of the corresponding Lax systems (in our case $t,\alpha$ and $\beta$). This extended $1$-form in turn possesses a remarkable Hamiltonian interpretation (see e.g. \cite{IP}, Appendix $A.1$ for the sine-Gordon equation, a simple transformation of \eqref{e:1}) and en route produces a formula in the style of \eqref{e:10}.\bigskip

The appearance of an action integral formula is thus no surprise\footnote{It is in fact already present in \cite{LZ} where $A\left(\frac{1}{\pi}\right)$ is computed from a field theoretic viewpoint.} and the point we are making in this article is that the general approach of \cite{IP} is also useful in the solution of the special 2D-Ising connection problem as it gives us an extremely simple way to compute $A(\lambda)$ (this constant cannot be obtained from the result of \cite{IP}, compare \cite{IP},(9) where we would require $\eta=0$ in order to match with \eqref{plus}). We refer the interested reader to \cite{BIP} where the same methodology has been used in the computation of tau-function expansions in random matrix theory.

\end{document}